\documentclass[11pt,letterpaper]{article}

\usepackage[margin=1in]{geometry}
\usepackage[T1]{fontenc}
\usepackage[utf8]{inputenc}
\usepackage{mathpazo}
\usepackage{microtype}
\usepackage{needspace}
\usepackage{amsmath,amssymb,amsthm,mathtools}
\usepackage[round,authoryear]{natbib}
\usepackage{xcolor}
\usepackage{aliascnt}
\usepackage{hyperref}
\usepackage[nameinlink,capitalize,noabbrev]{cleveref}

\hypersetup{
  pdftitle={Subquadratic Subsidies for Nonnegative or Nonpositive Valuations},
  pdfauthor={Max Dupr\'e la Tour and Mashbat Suzuki},
  colorlinks=true,
  linkcolor=blue!45!black,
  citecolor=blue!45!black,
  urlcolor=blue!45!black
}

\newtheorem{theorem}{Theorem}
\newaliascnt{lemma}{theorem}
\newtheorem{lemma}[lemma]{Lemma}
\aliascntresetthe{lemma}
\newaliascnt{proposition}{theorem}
\newtheorem{proposition}[proposition]{Proposition}
\aliascntresetthe{proposition}
\theoremstyle{remark}
\newaliascnt{remark}{theorem}

\aliascntresetthe{remark}

\newcommand{\R}{\mathbb{R}}
\newcommand{\E}{\mathbb{E}}
\newcommand{\Prb}{\mathbb{P}}
\newcommand{\defeq}{\coloneqq}
\newcommand{\calA}{\mathcal{A}}
\newcommand{\calB}{\mathcal{B}}
\newcommand{\wt}{\operatorname{wt}}

\title{Subquadratic Subsidies for Nonnegative or Nonpositive Valuations}
\author{Max Dupr\'e la Tour\thanks{RIKEN Center for Advanced Intelligence Project; The University of Tokyo; \texttt{maxduprelatour@gmail.com}} \and Mashbat Suzuki\thanks{UNSW Sydney; \texttt{mashbats@gmail.com}}}
\date{}

\begin{document}

\maketitle

\begin{abstract}
We study envy-freeness with subsidies for indivisible items beyond additive valuations. Assuming that every single-item marginal value lies in $[-1,1]$, we prove that a total subsidy of $O(n^{3/2}\sqrt{\log n})$ suffices to achieve envy-freeness among $n$ agents whenever all agents assign nonnegative values to every bundle or all assign nonpositive values to every bundle. These valuation classes include monotone goods and monotone chores, respectively, but do not require monotonicity. Our result establishes the first subquadratic total-subsidy bound for general monotone valuations that holds for every number of agents.
\end{abstract}

\section{Introduction}

How can we divide indivisible items fairly among agents with different preferences? An allocation is \emph{envy-free} if every agent values her own bundle at least as much as any other agent's bundle. Such an allocation need not exist: if two agents both value a single item positively, the agent who does not receive it envies the one who does. Monetary subsidies can compensate for these differences. If agent $i$ receives bundle $A_i$ and subsidy $p_i\geq0$, envy-freeness requires
\[
  v_i(A_i)+p_i\geq v_i(A_j)+p_j
  \qquad\text{for all agents }i,j.
\]
The subsidies are paid by an external source. We seek an allocation and subsidies that satisfy these inequalities at the smallest possible total cost.

Any bound on this cost must fix the scale of the valuations. For additive valuations, the usual assumption is that each item has absolute value at most one. For general valuations, we require that adding or removing a single item changes a bundle's value by at most one. We also normalize the value of the empty bundle to zero. The central question is then:
\begin{quote}
How much total subsidy suffices to guarantee envy-freeness in the worst case?
\end{quote}

\citet{HalpernShah2019} initiated the modern algorithmic study of this question. For additive goods, \citet{BrustleEtAl2020} proved that one unit of subsidy per agent suffices. Subtracting the smallest payment preserves envy-freeness and gives a total of at most $n-1$. This is tight: with one item valued at one by every agent, each of the $n-1$ nonrecipients needs at least one unit of subsidy. \citet{lu2026optimalsubsidyboundsgoods} extended this guarantee to mixed additive valuations, where an item may be a good for some agents and a chore for others.

For general monotone valuations, the best total-subsidy bounds that apply to every number of agents are substantially larger. \citet{BrustleEtAl2020} proved that a total subsidy of $2(n-1)^2$ suffices, and \citet{KawaseEtAl2024} improved this bound to $(n^2-n-1)/2$ for $n\geq 3$. \citet{LiuEtAl2024} asked whether a subquadratic bound is possible. A partial answer was given by \citet{DupreFujii2025}, who established a total-subsidy bound of $O(n^{3/2}\sqrt{\log n})$ for arbitrary valuations, but only when the number of agents is a prime power. Thus, obtaining a subquadratic bound for an arbitrary number of agents remained open.

\subsection{Our contribution}

We resolve this question by removing the prime-power restriction: for every number $n$ of agents, a total subsidy of $O(n^{3/2}\sqrt{\log n})$ suffices. Our result applies whenever all agents assign nonnegative values to every bundle or all assign nonpositive values to every bundle. These two classes include monotone goods and monotone chores, respectively, but neither requires monotonicity.

\begin{theorem}\label{thm:main}
Let $N=[n]$ be a set of agents and let $M$ be a finite set of indivisible items. Suppose that each valuation $v_i:2^M\to\R$ satisfies $v_i(\emptyset)=0$ and
\begin{equation}\label{eq:marginal-bound-intro}
  \bigl|v_i(S\cup\{g\})-v_i(S)\bigr|\leq1
  \qquad \forall S\subseteq M, \text{ and } \forall g\in M\setminus S.
\end{equation}
Assume that either $v_i(S)\geq0$ for every $i,S$, or $v_i(S)\leq0$ for every $i,S$. Then there exist an allocation $\calA=(A_1,\ldots,A_n)$ and nonnegative subsidies $(p_1,\ldots,p_n)$ that are envy-free and satisfy
\begin{equation}\label{eq:main-bound}
  \sum_{i\in N}p_i
  \leq 5\sqrt{2}\,(n-1)^{3/2}\sqrt{\log(4n)}.
\end{equation}
\end{theorem}

The proof also bounds each agent's subsidy by $5\sqrt{2(n-1)\log(4n)}$. The guarantee is existential: the argument identifies a suitable partition through a topological existence theorem, rather than a polynomial-time algorithm.

\subsection{Proof overview}

The proof combines supporting bundle prices, a topological equalization argument, and randomized rounding. The key is to equalize expected prices in a distribution that randomizes only a small number of items.

\paragraph{Bundles, prices, and allocations.}
Fix an unordered partition $\calB=\{B_1,\ldots,B_n\}$ of the items, allowing empty bundles. For matching purposes, give each bundle occurrence an arbitrary label $j\in[n]$. We call the corresponding matching position \emph{slot $j$}; its contents are the bundle $B_j$. These labels specify neither an order on the partition nor its recipients. A maximum-weight perfect matching $\sigma$ between agents and slots, with edge weights $v_i(B_j)$, gives the allocation $A_i=B_{\sigma(i)}$.

We seek bundle prices that support this matching:
\[
  v_i(B_{\sigma(i)})-q_{\sigma(i)}\geq v_i(B_k)-q_k
  \qquad\text{for all }i,k.
\]
At these prices, every agent receives a bundle that maximizes her value
minus its price, and all bundles are assigned. We use the coordinatewise
least nonnegative supporting price vector, which we call the
\emph{canonical prices}. For fixed slot labels, these prices depend
only on the partition, not on the choice of maximum-weight perfect
matching, and satisfy $\min_k q_k=0$.

Give every agent a common allowance $q_{\max}=\max_jq_j$ and charge the price of her assigned bundle. Her net subsidy is $p_i=q_{\max}-q_{\sigma(i)}\geq0$, and the resulting outcome is envy-free. At least one agent receives zero subsidy, so the total is at most $(n-1)q_{\max}$. It therefore suffices to find a partition with a small spread in its canonical prices.

The prices also have a classical welfare interpretation~\citep{GulStacchetti1999}. Let $W$ be the maximum welfare obtainable from the original bundles via maximum weight matching between agents and the bundles. Introduce an additional slot offering a formal copy of $B_k$, and let $W_k^+$ be the maximum welfare when all $n$ agents are matched to distinct slots, leaving one of the $n+1$ slots unused. This is an auxiliary matching problem, not a duplication of the actual items. We give a direct proof that the canonical prices satisfy
\[
  q_k=W_k^+-W.
\]
This identity controls how prices change when an item moves between bundles: $W$ changes by at most $2$ and $W_k^+$ by at most $3$, so each $q_k$ changes by at most $5$. 

\paragraph{Equalizing expected prices.}
Arrange the items as consecutive unit intervals and cut the resulting interval into $n$ \emph{pieces}. A piece is a geometric interval used to specify probabilities. Associate piece $k$ with slot $k$, and assign each whole item independently to a slot with probability equal to the fraction of its interval lying in the corresponding piece. The items assigned to slot $k$ form its random bundle. Only items cut in their interior have random destinations, so at most $n-1$ items are randomized.

Let $Q_k$ be the expected canonical price of the random bundle in slot $k$. These functions vary continuously with the cut positions. We use the KKM lemma to choose the cuts so that
\[
  Q_1=\cdots=Q_n.
\]
This is where the sign assumption enters. A zero-length piece produces an empty bundle with probability one. For nonnegative valuations, an empty bundle has price zero, so $Q_k=0$ whenever piece $k$ has length zero. This is the boundary condition needed for a direct application of KKM. For nonpositive valuations, an empty bundle has maximum price. We therefore apply the same argument to the nonnegative functions $\max_jQ_j-Q_k$, which again vanish when piece $k$ has length zero.

\paragraph{Rounding.}
Let $\mu$ be the common expected price, that is, $\mu=Q_i$  for each $i\in [n]$. Since at most $n-1$ items are randomized and each changes a price by at most $5$, McDiarmid's inequality and a union bound give a partition in which every price lies within $O(\sqrt{n\log n})$ of $\mu$. The price spread has the same order. Since the minimum canonical price is zero, the maximum price is also $O(\sqrt{n\log n})$, yielding the total-subsidy bound of $O(n\sqrt{n\log n})$.

\subsection{Related work}

\paragraph{Subsidies beyond additivity.}
Linear subsidy bounds are known for several restricted nonadditive classes. For matroid-rank valuations, \citet{GokoEtAl2024} obtained a total bound of $n-1$ and studied additional truthfulness and efficiency guarantees. \citet{BarmanEtAl2022} proved the same tight total bound for all dichotomous valuations, whose single-item marginals lie in $\{0,1\}$, without assuming submodularity. \citet{LiEtAl2025} established an $n-1$ bound for general monotone valuations in graphical allocation instances, where each item connects two agents, is valued only by those agents, and must be assigned to one of them. Our result places no such restrictions on the monotone valuations or on which agents may receive an item.

For broader valuation classes, \citet{KawaseEtAl2024} showed how to obtain an envy-free outcome with total subsidy at most $n(n-1)/2$ from an allocation that is envy-free up to one item. Their result yields this quadratic guarantee for doubly monotone valuations, where each item is consistently a good or a chore for each agent. In particular, it covers monotone chores, which are included in our nonpositive setting. The survey of \citet{LiuEtAl2024} gives a broader account of fair division with subsidies and explicitly asks for a subquadratic bound for monotone valuations in its Open Question~9.

\paragraph{Discrepancy and rounding.}
For a prime-power number $n$ of agents, \citet{DupreFujii2025} established a total-subsidy bound of $O(n^{3/2}\sqrt{\log n})$ for arbitrary valuations with single-item marginals in $[-1,1]$. Their proof relies on a discrepancy theorem for non-additive set functions, obtained by extending the constrained necklace-splitting argument of \citet{jojic2021splitting}. The prime-power restriction stems from this topological ingredient, leaving open whether a subquadratic subsidy bound holds for every number of agents, even for monotone goods. 

The two results are complementary: their subsidy theorem allows arbitrary signs but requires a prime-power number of agents, whereas ours uses a different topological argument to obtain the same asymptotic bound for every number of agents, provided that all agents' bundle values are nonnegative or all are nonpositive.

\paragraph{Topological fair division.}

The role of empty bundles has a close parallel in topological fair division. In \citet{Su1999}, the assumption that agents either always prefer nonempty pieces (hungry preferences) or never do so (lazy preferences) provides the boundary condition for the Sperner argument establishing an envy-free cake division into connected pieces. \citet{avvakumov2021envy} show that, without these preference assumptions, existence is still guaranteed when the number of agents is a prime power, but may fail otherwise. In our setting, the sign assumptions play an analogous role, supplying boundary conditions for expected prices rather than for preferences over cake pieces.

\paragraph{Assignment markets.}
For a fixed partition, the supporting-price problem is closely related to classical assignment markets and multi-item auctions~\citep{ShapleyShubik1971,DemangeGaleSotomayor1986}. We use the marginal-welfare characterization of minimum prices~\citep{GulStacchetti1999} to establish one-item stability, with a self-contained derivation for mandatory assignment and possibly negative bundle values.

\section{Model and preliminaries}\label{sec:preliminaries}

Let $N=[n]=\{1,\ldots,n\}$ be the agents and let $M=\{g_1,\ldots,g_m\}$ be a set of items. Each agent $i$ has a valuation $v_i:2^M\to\R$ with $v_i(\emptyset)=0$. Throughout, we assume the marginal bound
\begin{equation}\label{eq:marginal-bound}
  \bigl|v_i(S\cup\{g\})-v_i(S)\bigr|\leq1
  \qquad\text{for all }i\in N,\ S\subseteq M,\ g\notin S.
\end{equation}
We consider two classes of valuations: \emph{nonnegative valuations}, for which $v_i(S)\geq0$ for all $i,S$, and \emph{nonpositive valuations}, for which $v_i(S)\leq0$ for all $i,S$. Each class allows nonmonotone valuations. For example, with two items $a,b$, the valuation $v(\emptyset)=v(\{a,b\})=0$ and $v(\{a\})=v(\{b\})=1$ is nonnegative and has marginals in $[-1,1]$, but is not monotone. Monotone goods satisfy $v_i(S)\leq v_i(T)$ whenever $S\subseteq T$ and are nonnegative under our normalization. Monotone chores satisfy the reverse inequality and are nonpositive.

\paragraph{Bundles and partitions.}
A \emph{bundle} is a subset of $M$. A \emph{partition into $n$ bundles} is an unordered collection $\calB=\{B_1,\ldots,B_n\}$ of pairwise disjoint bundles whose union is $M$. We allow repeated empty bundles and count each occurrence separately, so the collection has $n$ members counted with multiplicity. Formally, this is a multiset; throughout, braces denote the unordered collection, including any repeated empty bundles.

\paragraph{Slots and labels.}
To represent a partition in a matching problem, choose an arbitrary labeling of its $n$ bundle occurrences by $[n]$. A \emph{slot} is one of these labeled matching positions: slot $j$ contains bundle $B_j$. Slots remain distinct even when their bundles are empty. The labels are auxiliary; they impose neither an order on the partition nor an assignment to agents. We keep them fixed when comparing or randomizing partitions. Indexed quantities such as $q_j(\calB)$ are understood relative to this chosen labeling, which we suppress in the notation. 

\paragraph{Matchings and allocations.}
Given a partition $\calB$ and its slot labels, a bijection $\sigma:N\to[n]$ matches agents to slots. It induces the \emph{allocation} $\calA=(A_i)_{i\in N}$ defined by
\[
  A_i=B_{\sigma(i)}\qquad\text{for each }i\in N.
\]
Thus $B_j$ is the bundle occupying slot $j$, whereas $A_i$ is the bundle received by agent $i$.  An \emph{empty slot} contains the empty bundle and can still be assigned to an agent; an \emph{unused slot} has no agent matched to it. In a perfect matching every slot is used, including the empty ones. For a fixed partition, we use a maximum-weight perfect matching with edge weights $v_i(B_j)$. The maximum welfare is
\begin{equation}\label{eq:welfare}
  W(\calB)\defeq\max_{\sigma\in S_n}\sum_{i\in N}v_i(B_{\sigma(i)}),
\end{equation}
where $S_n$ is the set of bijections between agents and bundle slots.

\paragraph{Prices and subsidies.}
Prices are indexed by slots: $q_j(\calB)$ is the price attached to the bundle $B_j$ in slot $j$. Subsidies are indexed by agents: $p_i\geq0$ is the payment received by agent $i$. An allocation $\calA$ and subsidies $p$ are \emph{envy-free} if
\begin{equation}\label{eq:ef}
  v_i(A_i)+p_i\geq v_i(A_\ell)+p_\ell
  \qquad\text{for every }i,\ell\in N.
\end{equation}
Given supporting prices and an optimal matching $\sigma$, we will use the allocation $A_i=B_{\sigma(i)}$ and payments $p_i=\max_jq_j(\calB)-q_{\sigma(i)}(\calB)$.

We use two standard tools. The first is the KKM theorem, which we state on the simplex
\[
  \Delta=\bigl\{x\in\R_{\geq0}^n:\textstyle\sum_kx_k=1\bigr\}.
\]
For nonempty $S\subseteq[n]$, let $\Delta_S=\{x\in\Delta:x_k=0\text{ for }k\notin S\}$ be the face spanned by $S$.

\begin{theorem}[KKM theorem \citep{KKM1929}]\label{thm:kkm}
Let $C_1,\ldots,C_n$ be closed subsets of $\Delta$. Suppose that $\Delta_S\subseteq\bigcup_{k\in S}C_k$ for every nonempty $S\subseteq[n]$. Then $\bigcap_{k\in[n]}C_k\neq\emptyset$.
\end{theorem}

The second tool bounds the deviation of a function of independent random choices when no single choice has much influence.

\begin{lemma}[McDiarmid's inequality \citep{McDiarmid1989}]\label{lem:mcdiarmid}
Let $X_1,\ldots,X_h$ be independent random variables, where $h\geq1$, and let $f=f(X_1,\ldots,X_h)$. Suppose that changing one coordinate changes $f$ by at most $L>0$. Then, for every $t>0$,
\begin{equation}\label{eq:mcdiarmid}
  \Prb\bigl[|f-\E f|\geq t\bigr]
  \leq 2\exp\!\left(-\frac{2t^2}{hL^2}\right).
\end{equation}
\end{lemma}

\section{Bundle prices}\label{sec:prices}

Fix a partition $\calB=\{B_1,\ldots,B_n\}$ and choose its slot labels. A price vector $\pi\in\R^n$ \emph{supports} $\calB$ if there is a matching $\sigma\in S_n$ such that
\begin{equation}\label{eq:supporting-prices}
  v_i(B_{\sigma(i)})-\pi_{\sigma(i)}\geq v_i(B_k)-\pi_k
  \qquad(i\in N,\ k\in[n]).
\end{equation}
At these prices, every agent receives a bundle that maximizes her value minus its price. Since every bundle is assigned, these are \emph{market-clearing prices} for $\calB$. Every agent must be matched to one slot and receives its bundle, which may be empty; leaving the market is not an option.

Any supported matching maximizes welfare: summing \cref{eq:supporting-prices} over agents and comparing with another matching cancels all prices. Moreover, prices supporting one optimal matching support every optimal matching. Each agent weakly prefers her bundle in the supported matching, but the utility differences sum to zero because both matchings have the same welfare and total price. Hence every agent is indifferent between her bundles in the two matchings. Supporting prices are therefore a property of the partition, independent of the optimal matching used to allocate it.

Only price differences matter, since adding the same constant to every coordinate preserves \cref{eq:supporting-prices}. We will use the coordinatewise least nonnegative supporting vector, whose entries are the \emph{canonical bundle prices} $q_j(\calB)$. To construct it, fix any maximum-weight matching $\sigma$ and write $a_j=\sigma^{-1}(j)$ for the agent matched to bundle $B_j$.

Define the complete directed \emph{weighted envy graph} $G_{\calB,\sigma}$ on the bundle slots~\citep{HalpernShah2019}, with arc weights
\begin{equation}\label{eq:envy-weight}
  w_{\calB,\sigma}(j,k)
  \defeq
  v_{a_j}(B_k)-v_{a_j}(B_j).
\end{equation}
Indexing vertices by slots is equivalent to indexing them by their assigned agents. A positive arc weight means that the agent assigned to the source envies the destination bundle; we retain zero and negative arc weights as well. Every directed cycle has nonpositive weight, since its weight is exactly the welfare gain from moving each agent on the cycle to the next bundle.

For each slot $k$, define
\begin{equation}\label{eq:q-path}
  q_k(\calB)
  \defeq
  \max\bigl\{\wt(P):P\text{ is a directed walk in }G_{\calB,\sigma}\text{ ending at }k\bigr\},
\end{equation}
where $\wt(P)$ is the sum of the arc weights and a walk of length zero has weight zero. The maximum is finite: deleting a repeated-vertex segment removes a nonpositive cycle and cannot decrease the weight. Hence a maximum is attained by a simple path.

\subsection{Canonical prices and subsidies}

The next lemma shows that $q$ gives the canonical prices and explains how to turn them into subsidies.

\begin{lemma}\label{lem:price-properties}
For every partition $\calB$ and every maximum-weight matching $\sigma$:
\begin{enumerate}
  \item $q(\calB)$ is the coordinatewise least nonnegative price vector supporting $\calB$. In particular, $q_k(\calB)\geq q_j(\calB)+w_{\calB,\sigma}(j,k)$ for all $j,k\in[n]$;
  \item $\min_kq_k(\calB)=0$;
  \item writing $q_{\max}(\calB)=\max_kq_k(\calB)$, the subsidies
  \[
    p_i=q_{\max}(\calB)-q_{\sigma(i)}(\calB)\qquad(i\in N)
  \]
  make the allocation $A_i=B_{\sigma(i)}$ envy-free with total subsidy at most $(n-1)q_{\max}(\calB)$.
\end{enumerate}
\end{lemma}

\begin{proof}
By \cref{eq:supporting-prices} and matching independence, a vector $x$ supports $\calB$ if and only if $x_k\geq x_j+w_{\calB,\sigma}(j,k)$ for all $j,k$. Appending the arc $(j,k)$ to a maximum-weight walk ending at $j$ shows that $q$ satisfies these inequalities.

To see that $q$ is the least such vector, let $x$ be any nonnegative supporting price vector. Summing the inequalities along any walk from $j$ to $k$ gives $x_k\geq x_j+\wt(P)\geq\wt(P)$, so $x$ dominates $q$. This proves the first claim. If $\min_kq_k>0$, subtracting this minimum from every coordinate of $q$ would give a strictly smaller nonnegative supporting vector, a contradiction.

Finally, give each agent a common allowance $q_{\max}$ and charge the price of her assigned bundle. The net payments $p_i=q_{\max}-q_{\sigma(i)}$ are nonnegative and satisfy
\[
  p_{a_j}-p_{a_k}
  =q_k-q_j
  \geq v_{a_j}(B_k)-v_{a_j}(B_j).
\]
Rearranging gives the envy-free inequality for agent $a_j$ comparing her bundle with $B_k$. Every payment is at most $q_{\max}$, and the agent matched to a slot attaining $q_{\max}$ receives zero payment, so the total is at most $(n-1)q_{\max}$.
\end{proof}

\subsection{Prices as marginal welfare gains}

The canonical price of a bundle equals the welfare gain from making an additional copy available. The identity below is the counterpart, for mandatory assignment, of the marginal-welfare characterization of minimum equilibrium prices~\citep[Theorem~4]{GulStacchetti1999}. We give a direct proof and then use it to bound price changes.

For $k\in[n]$, introduce an additional slot labeled $k'$, offering a formal copy of $B_k$. Every agent $i$ values the option at $k'$ at $v_i(B_k)$, exactly as at slot $k$. Let $W_k^+(\calB)$ be the maximum welfare of a matching that assigns all $n$ agents to distinct slots in $[n]\cup\{k'\}$. Exactly one slot is unused. This is an auxiliary matching problem: the original partition is unchanged, and no items are duplicated in the eventual allocation.

\begin{lemma}[Duplicate-bundle identity]\label{lem:duplicate}
For every partition $\calB$ and every $k\in[n]$,
\begin{equation}\label{eq:duplicate}
  q_k(\calB)=W_k^+(\calB)-W(\calB).
\end{equation}
In particular, with the slot labels fixed, the price $q_k(\calB)$ depends on the partition and the slot $k$, but not on the maximum-weight matching used to construct the envy graph.
\end{lemma}

\begin{proof}
Fix a maximum-weight matching $\sigma$, with agent $a_j=\sigma^{-1}(j)$ assigned to $B_j$. We prove the two inequalities separately.

Let $P=(j_0,\ldots,j_t)$ be a simple path ending at $j_t=k$. Move agent $a_k$ to slot $k'$, and for each $s<t$, move agent $a_{j_s}$ to slot $j_{s+1}$, where she receives $B_{j_{s+1}}$. Keep every other agent in place. This is a valid matching that leaves slot $j_0$ unused, and its welfare gain is $\wt(P)$. The construction also applies when $t=0$, with gain zero. Taking a maximum-weight path gives
\[
  W_k^+(\calB)-W(\calB)\geq q_k(\calB).
\]

For the reverse inequality, extend the price vector by setting $q_{k'}=q_k$. Since $q$ supports $\sigma$ and slots $k$ and $k'$ offer identical value and price, each agent's value minus price in any augmented matching is at most her value minus price under $\sigma$. Consider an optimal augmented matching and let $u\in[n]\cup\{k'\}$ be its unused slot. The prices of its assigned slots sum to $\sum_{j\in[n]}q_j+q_k-q_u$. Since $q_u\geq0$, summing the agents' utility inequalities gives
\[
  W_k^+(\calB)\leq W(\calB)+q_k-q_u\leq W(\calB)+q_k,
\]
which proves the identity.
\end{proof}

\subsection{Empty bundles and one-item changes}

We first record the only step in the proof that uses the sign of the valuations.

\begin{lemma}[Empty-bundle boundary]\label{lem:empty}
Suppose that $B_k=\emptyset$.
\begin{enumerate}
  \item For nonnegative valuations, $q_k(\calB)=0$.
  \item For nonpositive valuations, $q_k(\calB)=\max_jq_j(\calB)$.
\end{enumerate}
\end{lemma}

\begin{proof}
For nonnegative valuations, every outgoing arc from the empty slot has nonnegative weight:
\[
  w_{\calB,\sigma}(k,j)=v_{a_k}(B_j)-v_{a_k}(\emptyset)=v_{a_k}(B_j)\geq0.
\]
By \cref{lem:price-properties}, $q_j\geq q_k+w_{\calB,\sigma}(k,j)\geq q_k$ for every $j$. Thus $q_k$ is a minimum price and equals zero.

For nonpositive valuations, every incoming arc to the empty slot has nonnegative weight:
\[
  w_{\calB,\sigma}(j,k)=v_{a_j}(\emptyset)-v_{a_j}(B_j)=-v_{a_j}(B_j)\geq0.
\]
Hence $q_k\geq q_j+w_{\calB,\sigma}(j,k)\geq q_j$ for every $j$, as required.
\end{proof}

Next, we use the matching identity to control how much a price changes when one item moves. This bound holds regardless of the signs of the valuations.

\begin{lemma}[One-item stability]\label{lem:lipschitz}
If two partitions $\calB$ and $\calB'$ differ only by moving one item from one bundle to another, with slot labels preserved, then
\begin{equation}\label{eq:lipschitz}
  |q_k(\calB)-q_k(\calB')|\leq5
  \qquad\text{for every }k\in[n].
\end{equation}
\end{lemma}

\begin{proof}
Suppose that item $g$ moves from the bundle in slot $r$ to the bundle in slot $s$. For any fixed matching to the original slots, only the values of the agents assigned to $r$ and $s$ change. Each changes by at most one, by \cref{eq:marginal-bound}, so the matching's welfare changes by at most $2$. The same bound holds for the maximum over all matchings: an optimal matching for either partition is a feasible matching for the other. Thus
\[
  |W(\calB)-W(\calB')|\leq2.
\]

For a fixed matching to the augmented slots defining $W_k^+$, at most three used slots have changed values: slots $r$ and $s$, and slot $k'$ if $k\in\{r,s\}$. Each assigned value changes by at most one. Maximizing over these matchings gives
\[
  |W_k^+(\calB)-W_k^+(\calB')|\leq3.
\]
The identity in \cref{lem:duplicate} and the triangle inequality now give $|q_k(\calB)-q_k(\calB')|\leq3+2=5$.
\end{proof}

\section{Equalizing expected prices}\label{sec:equalization}

We now construct a distribution over partitions in which all canonical prices have the same expectation and at most $n-1$ items have random destinations. Limiting the randomness to these items will make the final bound independent of $m$.

\subsection{Partitions from interval cuts}

Arrange the items in an arbitrary order and represent $g_t$ by the unit interval $I_t=[t-1,t]$. A \emph{piece} is an interval in an auxiliary division of $[0,m]$, used only to specify assignment probabilities. Divide $[0,m]$ into $n$ consecutive pieces, allowing pieces of length zero, and label them from left to right by $k\in[n]$. Their lengths form the vector
\[
  x\in\Delta_m
  \defeq\bigl\{x\in\R_{\geq0}^n:\textstyle\sum_kx_k=m\bigr\}.
\]
Set $s_0=0$ and $s_k=\sum_{\ell=1}^kx_\ell$, so piece $k$ is $J_k(x)=[s_{k-1},s_k]$. Note that if $x_k=0$ then the length of 
$J_k(x)$ is zero.
Associate piece $J_k(x)$ with slot $k$. Assign each whole item $g_t$ independently to exactly one slot, choosing slot $k$ with probability
\begin{equation}\label{eq:rounding-probabilities}
  \alpha_{t,k}(x)
  \defeq\operatorname{length}\bigl(I_t\cap J_k(x)\bigr)
  =\max\{0,\min\{t,s_k\}-\max\{t-1,s_{k-1}\}\}.
\end{equation}
These probabilities are nonnegative and sum to one for each item. Write
\[
  B_k(x)\defeq\{g_t\in M:g_t\text{ is assigned to slot }k\},
  \qquad
  R(x)\defeq\{B_1(x),\ldots,B_n(x)\}.
\]
Thus $B_k(x)$ is the random \emph{bundle} corresponding to the geometric \emph{piece} $J_k(x)$: the piece determines the assignment probabilities, and the bundle contains the whole items assigned to its slot. The partition $R(x)$ is unordered; its auxiliary slot labels are retained when evaluating indexed prices. Define
\begin{equation}\label{eq:Q}
  Q_k(x)\defeq\E\bigl[q_k(R(x))\bigr].
\end{equation}
Here $Q_k(x)$ is the expected price of the bundle in slot $k$, not a price assigned to the interval $J_k(x)$. We evaluate valuations only on bundles of whole items and prices only on the resulting partitions. 

Call an item \emph{fractional} at $x$ if it has positive assignment probability to at least two slots. At least one of the cuts \(s_1,\ldots,s_{n-1}\) lies strictly inside its unit interval. Each cut is strictly inside at most one item, so there are at most $n-1$ fractional items. A cut at an integer endpoint does not make an item fractional, and several cuts inside one item still randomize only that item.

Each function $Q_k$ is continuous. To see this explicitly, for a map $b:[m]\to[n]$, set $B_j(b)=\{g_t:b(t)=j\}$ and $\calB(b)=\{B_1(b),\ldots,B_n(b)\}$, retaining the labels specified by $b$. Independence gives
\[
  Q_k(x)=\sum_{b:[m]\to[n]}q_k(\calB(b))
  \prod_{t=1}^m\alpha_{t,b(t)}(x).
\]
This is a finite sum of products of continuous functions, since every $\alpha_{t,k}$ is continuous by \cref{eq:rounding-probabilities}.

If $x_k=0$, then $J_k(x)$ has length zero, and $\alpha_{t,k}(x)=0$ for every item, so the corresponding bundle $B_k(x)$ is empty with probability one. A piece of positive length may also yield an empty bundle after rounding. For nonnegative valuations, \cref{lem:empty} gives $q_k(R(x))=0$  when $x_k=0$, so $Q_k(x)=0$ with probability one. For nonpositive valuations, the same lemma gives $q_k(R(x))\geq q_j(R(x))$  on this boundary, for every $j$. Taking expectations preserves these inequalities. We therefore have the boundary conditions
\begin{equation}\label{eq:Q-boundary}
  x_k=0\quad\Longrightarrow\quad
  \begin{cases}
    Q_k(x)=0, & \text{for nonnegative valuations},\\
    Q_k(x)=\max_jQ_j(x), & \text{for nonpositive valuations}.
  \end{cases}
\end{equation}

\subsection{A KKM equalization lemma}

The following consequence of KKM turns these boundary conditions into equal expectations.

\begin{lemma}[KKM equalization]\label{lem:kkm-equalization}
Let $F_1,\ldots,F_n:\Delta_m\to\R_{\geq0}$ be continuous, and suppose that $F_k(x)=0$ whenever $x_k=0$. Then there is an $x^*\in\Delta_m$ such that
\[
  F_1(x^*)=\cdots=F_n(x^*).
\]
\end{lemma}

\begin{proof}
If $m=0$, the domain consists of the zero vector and all functions vanish there. Suppose that $m>0$. For each $k$, define the closed set
\[
  C_k=\bigl\{x\in\Delta_m:F_k(x)=\max_jF_j(x)\bigr\}.
\]
Consider the face supported on a nonempty set $S\subseteq[n]$.
On this face, $F_k=0$ for every $k\notin S$. Since all functions
are nonnegative, their maximum at each point is attained by some
index in $S$, even when all functions vanish. Hence the face is
covered by $\bigcup_{k\in S}C_k$. Rescaling $\Delta_m$ to the unit
simplex $\Delta$ and applying the KKM theorem (\cref{thm:kkm})
yields a point in $\bigcap_{k\in[n]}C_k$, at which
$F_1=\cdots=F_n$.
\end{proof}

\begin{proposition}[Equal expected prices]\label{prop:fractional}
There is an $x^*\in\Delta_m$ such that
\[
  Q_1(x^*)=\cdots=Q_n(x^*).
\]
In $R(x^*)$, at most $n-1$ items have random bundle membership, independently of one another; every other item belongs to a fixed bundle.
\end{proposition}

\begin{proof}
For nonnegative valuations, the functions $Q_k$ are nonnegative, continuous, and zero when $x_k=0$. Apply \cref{lem:kkm-equalization} with $F_k=Q_k$.

For nonpositive valuations, set
\[
  F_k(x)=\max_jQ_j(x)-Q_k(x).
\]
These functions are continuous and nonnegative, and \cref{eq:Q-boundary} gives $F_k(x)=0$ whenever $x_k=0$. By \cref{lem:kkm-equalization}, they are all equal at some $x^*$. At every point, at least one $F_k$ equals zero, so their common value at $x^*$ is zero. Therefore all $Q_k(x^*)$ are equal.

The statement about random items follows from the interval construction.
\end{proof}

\section{Rounding and subsidies}\label{sec:rounding}

We now draw an integral partition from the distribution in \cref{prop:fractional}. Concentration around the common expected price bounds the spread of the realized prices. Their minimum is zero, so the same bound controls the maximum price and hence the subsidies.

\begin{proof}[Proof of \cref{thm:main}]
If $n=1$, assign every item to the only agent and pay no subsidy. If $M=\emptyset$, give every agent the empty bundle and again pay no subsidy. Hence assume that $n\geq2$ and $m\geq1$.

Choose $x^*$ as in \cref{prop:fractional}, and write
\[
  \mu=Q_1(x^*)=\cdots=Q_n(x^*).
\]
Let $H$ be the set of fractional items and $h=|H|\leq n-1$. If $h=0$, the partition is deterministic and all its prices equal $\mu$. By \cref{lem:price-properties}, one price is zero, and hence they all are. A maximum-weight matching then gives an envy-free allocation with zero subsidies.

Now suppose that $h\geq1$. The price $q_k(R(x^*))$ is a function of the independent destinations of the $h$ fractional items. Changing one destination moves one item between bundles, so \cref{lem:lipschitz} bounds the change in this function by $5$. McDiarmid's inequality gives
\begin{equation}\label{eq:q-concentration}
  \Prb\bigl[|q_k(R(x^*))-\mu|\geq t\bigr]
  \leq2\exp\!\left(-\frac{2t^2}{25h}\right).
\end{equation}
Set
\[
  t=5\sqrt{\frac{h}{2}\log(4n)}.
\]
For each $k$, the probability in \cref{eq:q-concentration} is at most $1/(2n)$. A union bound shows that, with probability at least $1/2$, all $n$ prices differ from $\mu$ by less than $t$. Fix a realized partition $\widetilde{\calB}=\{\widetilde{B}_1,\ldots,\widetilde{B}_n\}$ with this property, where $\widetilde{B}_k$ is the realized bundle in slot $k$.

Since the minimum canonical price is zero by \cref{lem:price-properties}, there exists $j\in [n]$ such that  $q_j(\widetilde{\calB})=0$. Hence, $|q_j(\widetilde{\calB}) - \mu |= \mu \leq t$. As every canonical price is  within $t$ of $\mu$, it follows that
\[
  q_{\max}(\widetilde{\calB})
  \leq2t
  =5\sqrt{2h\log(4n)}
  \leq5\sqrt{2(n-1)\log(4n)}.
\]
Finally, choose a maximum-weight matching $\sigma$ for $\widetilde{\calB}$. Form the allocation $A_i=\widetilde{B}_{\sigma(i)}$ and pay agent $i$ the subsidy $p_i=q_{\max}(\widetilde{\calB})-q_{\sigma(i)}(\widetilde{\calB})$. By \cref{lem:price-properties}, the resulting outcome is envy-free, each agent receives at most $q_{\max}(\widetilde{\calB})$, and
\[
  \sum_{i\in N}p_i
  \leq(n-1)q_{\max}(\widetilde{\calB})
  \leq5\sqrt{2}\,(n-1)^{3/2}\sqrt{\log(4n)}.
\]
This proves the theorem.
\end{proof}

\section{Discussion}\label{sec:discussion}
We have prioritized a simple, unified analysis over optimizing the
constants in the subsidy bound. A more careful accounting in our
proof yields a total-subsidy bound of
\(
  (n-1)\sqrt{(2n+3)\log n}
\)
when valuations are either all nonnegative or all nonpositive.
For monotone goods or monotone chores, this bound improves to
\(
  (n-1)\sqrt{\frac{n+5}{2}\log n}.
\)

\paragraph{Beyond the sign assumption.}
The sign assumption enters our proof only through the empty-bundle
boundary condition: an empty bundle has minimum canonical price for
nonnegative valuations and maximum canonical price for nonpositive
valuations. This boundary condition enables the KKM argument to
equalize expected prices. All other ingredients---the matching
identity, one-item stability, and concentration---apply to arbitrary
valuations with bounded single-item marginals. Extending our approach
to mixed-sign valuations for every number of agents therefore calls
for an alternative way to equalize expected prices.

\paragraph{Closing the gap.}
The single-good example gives a lower bound of $n-1$, leaving a gap
between this linear lower bound and our upper bound of
$O(n^{3/2}\sqrt{\log n})$. In our proof, the factor $\sqrt{\log n}$
arises from the union bound over the $n$ prices. Can a different
rounding argument remove this factor? More broadly, does a linear
total-subsidy bound hold for general monotone valuations?

\paragraph{Efficient computation.}

A further question is whether our total-subsidy bound of
$O(n^{3/2}\sqrt{\log n})$ can be achieved by a polynomial-time
algorithm. For a fixed partition, the canonical prices can be
computed efficiently from the agent--bundle values by solving the
matching problems in \cref{lem:duplicate}. The challenge lies in
finding a suitable partition. Our KKM argument establishes the
existence of an equalizing point $x^*$ but does not provide a
polynomial-time procedure for finding it. Moreover, direct
evaluation of the expected-price functions may require summing
over exponentially many rounding outcomes. An efficient
implementation of our approach would need to overcome both
obstacles.

\bigskip

\noindent\textbf{Acknowledgments.} Mashbat Suzuki is supported by the ARC Laureate Project FL200100204 on ``Trustworthy AI''.

\bigskip

\noindent\textbf{Declaration of generative AI use.}
All research results in this paper were obtained by the authors
without AI assistance. GPT-5.6 Sol was used solely to refine the
exposition, improving the clarity, phrasing, and presentation
of the manuscript.

\begingroup
\small
\interlinepenalty=10000
\bibliographystyle{plainnat}
\bibliography{subquadratic-subsidies-nonnegative-nonpositive}

\begin{thebibliography}{17}
\providecommand{\natexlab}[1]{#1}
\providecommand{\url}[1]{\texttt{#1}}
\expandafter\ifx\csname urlstyle\endcsname\relax
  \providecommand{\doi}[1]{doi: #1}\else
  \providecommand{\doi}{doi: \begingroup \urlstyle{rm}\Url}\fi

\bibitem[Avvakumov and Karasev(2021)]{avvakumov2021envy}
Sergey Avvakumov and Roman Karasev.
\newblock Envy-free division using mapping degree.
\newblock \emph{Mathematika}, 67\penalty0 (1):\penalty0 36--53, 2021.
\newblock \doi{10.1112/mtk.12059}.

\bibitem[Barman et~al.(2022)Barman, Krishna, Narahari, and
  Sadhukhan]{BarmanEtAl2022}
Siddharth Barman, Anand Krishna, Yadati Narahari, and Soumyarup Sadhukhan.
\newblock Achieving envy-freeness with limited subsidies under dichotomous
  valuations.
\newblock In \emph{Proceedings of the Thirty-First International Joint
  Conference on Artificial Intelligence}, pages 60--66. International Joint
  Conferences on Artificial Intelligence Organization, 2022.
\newblock \doi{10.24963/ijcai.2022/9}.
\newblock URL \url{https://www.ijcai.org/proceedings/2022/9}.

\bibitem[Brustle et~al.(2020)Brustle, Dippel, Narayan, Suzuki, and
  Vetta]{BrustleEtAl2020}
Johannes Brustle, Jack Dippel, Vishnu~V. Narayan, Mashbat Suzuki, and Adrian
  Vetta.
\newblock One dollar each eliminates envy.
\newblock In \emph{Proceedings of the 21st ACM Conference on Economics and
  Computation}, pages 23--39. ACM, 2020.
\newblock \doi{10.1145/3391403.3399447}.

\bibitem[Demange et~al.(1986)Demange, Gale, and
  Sotomayor]{DemangeGaleSotomayor1986}
Gabrielle Demange, David Gale, and Marilda Sotomayor.
\newblock Multi-item auctions.
\newblock \emph{Journal of Political Economy}, 94\penalty0 (4):\penalty0
  863--872, 1986.
\newblock \doi{10.1086/261411}.

\bibitem[{Dupr{\'e} la Tour} and Fujii(2025)]{DupreFujii2025}
Max {Dupr{\'e} la Tour} and Kaito Fujii.
\newblock Discrepancy and fair division for non-additive valuations.
\newblock \emph{arXiv preprint arXiv:2509.16802}, 2025.
\newblock \doi{10.48550/arXiv.2509.16802}.
\newblock URL \url{https://arxiv.org/abs/2509.16802}.

\bibitem[Goko et~al.(2024)Goko, Igarashi, Kawase, Makino, Sumita, Tamura,
  Yokoi, and Yokoo]{GokoEtAl2024}
Hiromichi Goko, Ayumi Igarashi, Yasushi Kawase, Kazuhisa Makino, Hanna Sumita,
  Akihisa Tamura, Yu~Yokoi, and Makoto Yokoo.
\newblock A fair and truthful mechanism with limited subsidy.
\newblock \emph{Games and Economic Behavior}, 144:\penalty0 49--70, 2024.
\newblock \doi{10.1016/j.geb.2023.12.006}.

\bibitem[Gul and Stacchetti(1999)]{GulStacchetti1999}
Faruk Gul and Ennio Stacchetti.
\newblock {Walrasian} equilibrium with gross substitutes.
\newblock \emph{Journal of Economic Theory}, 87\penalty0 (1):\penalty0 95--124,
  1999.
\newblock \doi{10.1006/jeth.1999.2531}.

\bibitem[Halpern and Shah(2019)]{HalpernShah2019}
Daniel Halpern and Nisarg Shah.
\newblock Fair division with subsidy.
\newblock In \emph{Algorithmic Game Theory: 12th International Symposium, SAGT
  2019}, pages 374--389. Springer, 2019.
\newblock \doi{10.1007/978-3-030-30473-7_25}.

\bibitem[Joji{\'c} et~al.(2021)Joji{\'c}, Panina, and
  {\v{Z}}ivaljevi{\'c}]{jojic2021splitting}
Du{\v{s}}ko Joji{\'c}, Gaiane Panina, and Rade {\v{Z}}ivaljevi{\'c}.
\newblock Splitting necklaces, with constraints.
\newblock \emph{SIAM Journal on Discrete Mathematics}, 35\penalty0
  (2):\penalty0 1268--1286, 2021.
\newblock \doi{10.1137/20M1331949}.

\bibitem[Kawase et~al.(2024)Kawase, Makino, Sumita, Tamura, and
  Yokoo]{KawaseEtAl2024}
Yasushi Kawase, Kazuhisa Makino, Hanna Sumita, Akihisa Tamura, and Makoto
  Yokoo.
\newblock Towards optimal subsidy bounds for envy-freeable allocations.
\newblock \emph{Proceedings of the AAAI Conference on Artificial Intelligence},
  38\penalty0 (9):\penalty0 9824--9831, 2024.
\newblock \doi{10.1609/aaai.v38i9.28842}.

\bibitem[Knaster et~al.(1929)Knaster, Kuratowski, and Mazurkiewicz]{KKM1929}
Bronis{\l}aw Knaster, Kazimierz Kuratowski, and Stefan Mazurkiewicz.
\newblock {Ein Beweis des Fixpunktsatzes f{\"u}r $n$-dimensionale Simplexe}.
\newblock \emph{Fundamenta Mathematicae}, 14\penalty0 (1):\penalty0 132--137,
  1929.

\bibitem[Li et~al.(2025)Li, Sun, Suzuki, and Xing]{LiEtAl2025}
Bo~Li, Ankang Sun, Mashbat Suzuki, and Shiji Xing.
\newblock On the subsidy of envy-free orientations in graphs.
\newblock arXiv preprint arXiv:2502.13671, 2025.
\newblock URL \url{https://arxiv.org/abs/2502.13671}.

\bibitem[Liu et~al.(2024)Liu, Lu, Suzuki, and Walsh]{LiuEtAl2024}
Shengxin Liu, Xinhang Lu, Mashbat Suzuki, and Toby Walsh.
\newblock Mixed fair division: A survey.
\newblock \emph{Journal of Artificial Intelligence Research}, 80:\penalty0
  1373--1406, 2024.
\newblock \doi{10.1613/jair.1.15800}.
\newblock URL \url{https://jair.org/index.php/jair/article/view/15800}.

\bibitem[Lu et~al.(2026)Lu, Mackenzie, and
  Suzuki]{lu2026optimalsubsidyboundsgoods}
Xinhang Lu, Simon Mackenzie, and Mashbat Suzuki.
\newblock Optimal subsidy bounds for goods and chores: One dollar each
  suffices.
\newblock arXiv preprint arXiv:2607.10089, 2026.
\newblock URL \url{https://arxiv.org/abs/2607.10089}.

\bibitem[McDiarmid(1989)]{McDiarmid1989}
Colin McDiarmid.
\newblock On the method of bounded differences.
\newblock In \emph{Surveys in Combinatorics, 1989}, volume 141 of \emph{London
  Mathematical Society Lecture Note Series}, pages 148--188. Cambridge
  University Press, 1989.

\bibitem[Shapley and Shubik(1971)]{ShapleyShubik1971}
Lloyd~S. Shapley and Martin Shubik.
\newblock The assignment game {I}: The core.
\newblock \emph{International Journal of Game Theory}, 1\penalty0 (1):\penalty0
  111--130, 1971.
\newblock \doi{10.1007/BF01753437}.

\bibitem[Su(1999)]{Su1999}
Francis~Edward Su.
\newblock Rental harmony: {Sperner}'s lemma in fair division.
\newblock \emph{The American Mathematical Monthly}, 106\penalty0 (10):\penalty0
  930--942, 1999.
\newblock \doi{10.2307/2589747}.

\end{thebibliography}
\endgroup

\end{document}